\documentclass[a4paper,onecolumn,11pt,accepted=2020-06-10]{quantumarticle}
\pdfoutput=1
\usepackage[utf8]{inputenc}
\usepackage[english]{babel}
\usepackage[T1]{fontenc}
\usepackage{amsmath}
\usepackage{hyperref}

\usepackage[numbers,sort&compress]{natbib}

\usepackage{tikz}
\usepackage{lipsum}
\usepackage{IEEEtrantools}
\usepackage{amstext}
\usepackage{amsfonts}
\usepackage{amsthm}
\usepackage{bbold}
\usepackage{mathtools}
\usepackage{scrextend}
\usepackage{array}
\usepackage{multicol}
\usepackage{mathrsfs}
\usepackage{color}
\usepackage{hyperref}
\usepackage{bm}

\newtheorem{theorem}{Theorem}[section]    
\newtheorem{corollary}[theorem]{Corollary}    
\newtheorem{proposition}[theorem]{Proposition}    
\newtheorem{lemma}[theorem]{Lemma}    

\renewcommand{\qed}{\hfill{$\rule{6pt}{6pt}$}} 
\renewenvironment{proof}{\noindent{\bf Proof:}}{\qed\\}
\newenvironment{proofof}[1]{\noindent{\bf Proof of #1:}}{\qed\\}

\numberwithin{equation}{section}

\newcommand{\complex}{{\mathbb C}}
\newcommand{\reals}{{\mathbb R}}

\newcommand{\naturals}{{\mathbb N}}

\newcommand{\ket}[1]{| #1 \rangle}
\newcommand{\bra}[1]{ \langle #1 |}
\newcommand{\ketbra}[2]{| #1 \rangle\!\langle #2 |}

\newcommand{\norm}[1]{\left\| #1 \right\|}
\newcommand{\trnorm}[1]{\left\| #1 \right\|_{\mathrm{tr}}}
\newcommand{\size}[1]{\left| #1 \right|}
\newcommand{\set}[1]{\left\{ #1 \right\}}

\newcommand{\ceil}[1]{{\lceil #1 \rceil}}
\newcommand{\trace}{{\mathrm{Tr}}}
\newcommand{\support}{{\mathrm{supp}}}
\newcommand{\Order}{{\mathrm{O}}}
\newcommand{\order}{{\mathrm{o}}}
\newcommand{\density}[1]{\ketbra{#1}{#1}}

\newcommand{\ancilla}{\ket{\bar{0}}}
\newcommand{\e}{{\mathrm{e}}}

\newcommand{\expct}{{\mathbb E}}
\newcommand{\abs}[1]{\left| #1 \right|}

\newcommand{\id}{{\mathbb 1}}

\newcommand{\linear}{{\mathsf L}}
\newcommand{\unitary}{{\mathsf U}}
\newcommand{\sphere}{\mathrm{Sphere}}
\newcommand{\qstate}{{\mathsf D}}

\newcommand{\eqdef}{\coloneqq}
\newcommand{\tensor}{\otimes}

\newcommand{\suppress}[1]{}
\newcommand{\comment}[1]{}

\newcommand{\rS}{{\mathrm S}}
\newcommand{\rI}{{\mathrm I}}

\newcommand{\rF}{{\mathrm F}}

\newcommand{\Imax}{{\mathrm I}_{\max}}
\newcommand{\Smax}{{\mathrm S}_{\max}}
\newcommand{\Smin}{{\mathrm S}_{\min}}

\newcommand{\rP}{{\mathrm P}}

\newcommand{\QIC}{{\mathrm {QIC}}}

\newcommand{\Pos}{{\mathsf{Pos}}}

\newcommand{\sB}{{\mathsf B}}

\newcommand{\cK}{{\mathcal K}}
\newcommand{\cH}{{\mathcal H}}
\newcommand{\cM}{{\mathcal M}}
\newcommand{\cX}{{\mathcal X}}
\newcommand{\cY}{{\mathcal Y}}

\newcommand{\bU}{{\bm U}}
\newcommand{\bcZ}{{\bm{\mathcal Z}}}
\newcommand{\bM}{{\bm M}}
\newcommand{\bB}{{\bm B}}
\newcommand{\bE}{{\bm E}}
\newcommand{\bI}{{\bm I}}
\newcommand{\cA}{{\mathcal A}}
\newcommand{\cB}{{\mathcal B}}
\newcommand{\cC}{{\mathcal C}}
\newcommand{\cW}{{\mathcal W}}
\newcommand{\fT}{{\mathfrak T}}

\newcommand{\rA}{{\mathrm A}}
\newcommand{\rB}{{\mathrm B}}
\newcommand{\out}{{\mathrm{out}}}
\newcommand{\rin}{{\mathrm{in}}}

\begin{document}

\title{On the Entanglement Cost of One-Shot Compression}

\author[1]{Shima Bab Hadiashar}
\orcid{0000-0001-6707-6438}
\email{sbabhadi@uwaterloo.ca}
\author[1]{Ashwin Nayak}
\email{ashwin.nayak@uwaterloo.ca}
\orcid{0000-0001-9866-9316}
\affiliation[1]{Department of Combinatorics and Optimization,
and Institute for Quantum Computing, University
of Waterloo, 200 University Ave.\ W., Waterloo, ON,
N2L~3G1, Canada.}

\maketitle

\begin{abstract}
  We revisit the task of visible compression of an ensemble of quantum
states with entanglement assistance in the one-shot setting. The
protocols achieving the best compression use many more qubits of shared
entanglement than the number of qubits in the states in the ensemble. Other
compression protocols, with potentially larger communication cost, have
entanglement cost bounded by the number of qubits in the given states.  This
motivates the question as to whether entanglement is truly necessary for
compression, and if so, how much of it is needed.

Motivated by questions in communication complexity, we lift certain
restrictions that are placed on compression protocols in tasks such as
state-splitting and channel simulation. We show that an ensemble of the
form designed by Jain, Radhakrishnan, and Sen (ICALP'03) saturates the
known bounds on the sum of communication and entanglement costs, even
with the relaxed compression protocols we study.

The ensemble and the associated one-way communication protocol have 
several remarkable properties. The ensemble is incompressible by
more than a constant number of qubits without shared entanglement, even when
constant error is allowed. Moreover,
in the presence of shared entanglement, the communication cost of compression 
can be arbitrarily smaller than the entanglement cost. The quantum 
information cost of the protocol can thus be arbitrarily smaller
than the cost of compression without shared entanglement. The ensemble can also
be used to show the impossibility of reducing, via compression, the
shared entanglement used in two-party protocols for computing Boolean
functions.
\end{abstract}

\section{Introduction}

\subsection{Visible compression}

Compression of quantum states is a fundamental task in information
processing. In the simplest setting, we have two spatially separated
parties, commonly called Alice and Bob, who can communicate with each
other by exchanging quantum states. They have in mind an ensemble
of~$m$-dimensional quantum states
\begin{equation}
\label{eq-ensemble-intro}
\left( (p_x, \rho_x) : x \in S, ~ \rho_x \in \qstate(\complex^m) 
    \right) \enspace,
\end{equation}
where~$S$ is some
non-empty finite set, and~$p$ is a probability distribution over~$S$.
Alice gets an input~$x \in S$ with probability~$p_x$, and 
would like to send a message, i.e., a quantum state~$\sigma_x \in 
\qstate( \complex^d )$ to Bob so that he can recover the state~$\rho_x$,
or even an approximation to it.  Since the input~$x$ completely
specifies the corresponding state~$\rho_x$, this variant of the task
is called \emph{visible\/} compression.  The \emph{communication cost\/}
of the protocol is~$\log d$, the length of the message in qubits.
Their goal is to accomplish this with as short a message as possible, 
i.e., to minimize the dimension~$d$. A central question in quantum 
information theory is whether there is a simple characterization of 
the optimal communication cost in terms of the ``information content'' 
of the ensemble.

An additional resource that Alice and Bob may use in compression 
is a shared entangled state. In other words, the two parties may start
with their qubits initialized to a fixed pure quantum state independent of the
input received by Alice. The local quantum operations performed
for compression and decompression then also involve the respective parts
of the shared state. This is depicted in Figure~\ref{fig-compression},
and the protocol (or channel) is said to be \emph{with shared 
entanglement\/} or \emph{entanglement assisted\/}.
As we may expect, the communication cost may decrease due to the availability
of this additional resource. The \emph{entanglement cost\/} of a protocol
is the minimal dimension of the support of either party's share of the initial
state (measured in qubits) required to achieve some communication cost.
(We discuss the notion of entanglement cost in detail
in Section~\ref{sec-concl}.)
We would also like to characterize the entanglement cost in this setting, 
in addition to the communication cost.

\begin{figure}[h]
\centering
\includegraphics[width=400pt]{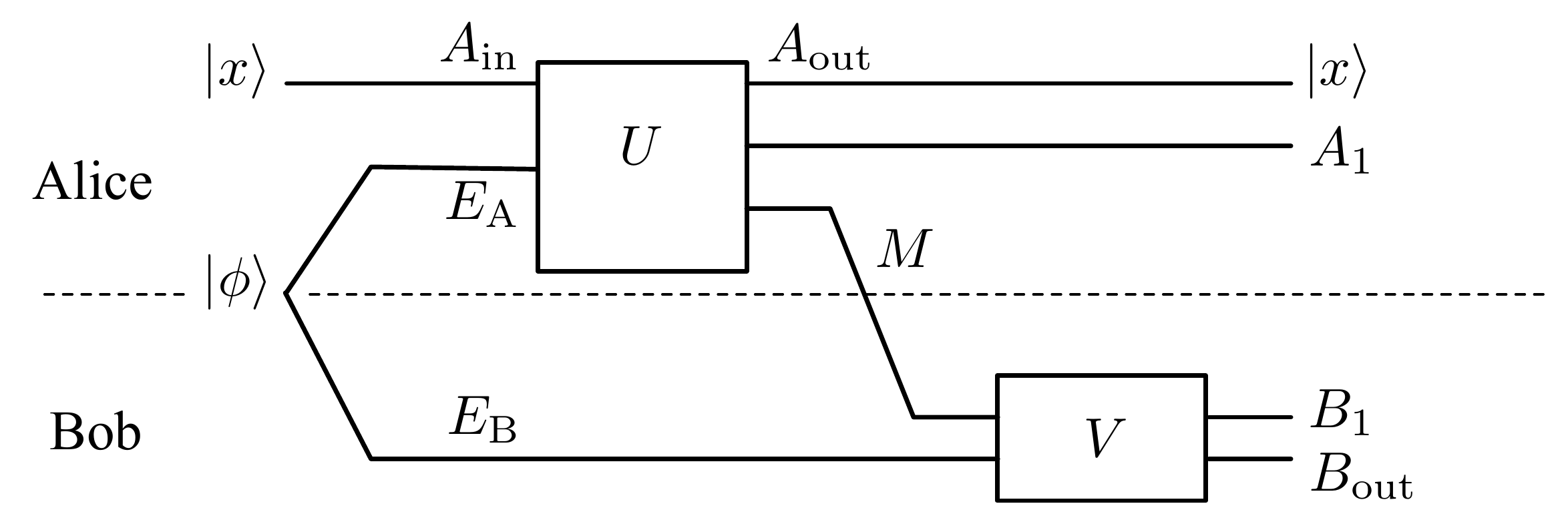}
\caption{A one-message protocol for compression of quantum states, with
shared entanglement. The register~$A_\rin$ holds the input given to Alice,
and~$E_\rA$ contains Alice's workspace and her part of the initial shared 
state (the shared entanglement). The register~$E_\rB$ contains Bob's
workspace and his part of the initial shared state. The compression is 
implemented by the isometry~$U$, and the register~$M$ contains the 
compressed state and is sent as the message. The decompression is 
implemented by the isometry~$V$. Bob's output is contained in the 
register~$B_\out$.}
\label{fig-compression}
\end{figure}

Compression problems similar to the one above have been studied
extensively in quantum information theory, both in the \emph{one-shot\/}
setting (the one we described above), and in the \emph{asymptotic
setting\/} (where the sender's input consists of multiple samples picked 
independently from the same distribution). The problem has been studied
in early works such as Ref.~\cite{BCFJS01-compression} in the setting of
quantum communication without shared entanglement. It is known as
\emph{remote state preparation\/} when allowed one-way communication
over a classical channel with shared entanglement. We refer the reader
to Ref.~\cite[Table~I]{BNR18} for a summary of the work on remote state
preparation; we describe the most relevant results---in the one-shot
setting---below.

Other tasks in the literature that come close to the one above are 
\emph{state splitting\/} (see, e.g., Ref.~\cite{BCR11-reverse-Shannon}),
and that of channel simulation in the context of the \emph{Quantum
Reverse Shannon
Theorem\/}~\cite{BDHSW14-reverse-Shannon,BCR11-reverse-Shannon}. State
splitting is the time reversal~\cite{Dev06-Tri-comm-protocol} of
\emph{state merging\/}~\cite{HOW05-state-merging,HOW07-state-merging},
and was called the ``fully quantum reverse Shannon protocol'' in
Ref.~\cite{Dev06-Tri-comm-protocol}.  We explain the connection to state
splitting in detail in Section~\ref{sec-compression}.

In both state splitting and channel simulation, the protocol is required
to be ``coherent'' in specific ways. In particular, in compressing an
ensemble of states as in Eq.~(\ref{eq-ensemble-intro}), at the end of
the protocol, Bob would be required to hold an approximation to the
state~$\rho_x$ and Alice a purification of this state. In contrast to
these tasks, we do not require that the compression protocol maintain
such coherence. More precisely, the registers containing a purification
of the output state may be shared by Alice and Bob. Such compression
protocols are more relevant in the context of two-party communication
protocols studied in complexity theory, especially in the context of
\emph{direct sum\/} and \emph{direct product\/} results (see e.g.,
Refs.~\cite{JRS03-direct-sum, JRS08-direct-sum, Touchette15-QIC} and the
references therein).  In communication complexity, a typical goal is to
compute a bivariate Boolean function when the inputs are distributed
between two parties. The parties communicate with each other,
alternating messages with local computation, and at the end, one party
produces the output of the protocol from the part of the final state in
her possession. As a result, the output of the protocol does not depend
on the part of the state held by the other party (i.e., on the
purification of her part of the final joint state). A compression scheme
for the final state then need only focus on the part being measured for
the output.

\subsection{Entanglement cost of compression}

Jain, Radhakrishnan, and Sen~\cite{JRS05-compression,JRS08-direct-sum}
gave a one-shot protocol for compressing an ensemble of states as in
Eq.~(\ref{eq-ensemble-intro}), and bounded its communication cost by
$\Order( \rI(A : B)_\tau / \epsilon^3)$, where~$\rI(A : B)_\tau$ is the
mutual information between registers~$A$ and~$B$ in the state~$\tau^{AB}
\coloneqq \tfrac{1}{n} \sum_{x \in [n]} \density{x}^A \tensor \rho_x^B$,
and~$\epsilon$ is the average approximation error (cf.\
Section~\ref{sec-compression} for a precise definition of average
error). Using a more refined application of their technique, Bab
Hadiashar, Nayak, and Renner~\cite{BNR18} tightly characterized the
communication cost of the task in terms of the \emph{smooth
max-information\/}, a one-shot entropic analogue of mutual
information. Their results are stated for entanglement-assisted
classical channels and use purified distance to quantify the
approximation, but translate immediately to the setting here
through the use of superdense coding~\cite[Section~6.3.1]{W18-TQI} and
the Fuchs and van de Graaf Inequalities (Proposition~\ref{prop-FvdG}). The
upper bound so obtained is
\[
\frac{1}{2} \, \Imax^{\epsilon/\sqrt{2}}(A:B)_{\tau} + \Order(\log \log (
    1/ \epsilon)) \enspace.
\]
This is slightly better than that derived from
protocols for state splitting in terms of the approximation error; it
has an additive term of~$\Order(\log \log \tfrac{1}{\epsilon})$ for
average error~$\epsilon$ versus the additive term of~$\Order(\log
\tfrac{1}{\epsilon})$ in
Ref.~\cite[Corollary~5]{AJ18-efficient-convex-split}. However, both these
protocols use shared entanglement that may be much longer than the
message itself, namely~$\Order(k (\log \tfrac{1}{\epsilon}) \log
m)$ qubits and~$\Order( (1 + 1/\epsilon^2) \log_2 (m / \epsilon))$
qubits, respectively, where~$\log_2 k = \rI(A:B)_\tau$, and~$m$ is the
dimension of the states in the ensemble. On the other hand, earlier
protocols for state splitting~\cite[Lemma~3.5]{BCR11-reverse-Shannon},
with potentially larger communication cost, have entanglement cost
bounded by~$\log m$. Since sharing entanglement also entails
some communication, in addition to the preparation and storage of a
potentially delicate high dimensional state, this motivates the question
as to whether shared entanglement is truly necessary for compression,
and if so, how much of it is needed.

For the more restrictive task of state splitting, it follows from the
proof of the converse bound for one-shot entanglement consumption due to
Berta, Christandl, and
Touchette~\cite[Proposition~10]{BCT16-state-redistribution} that the sum
of the communication and entanglement costs is at least the
\emph{min-entropy\/}~$\Smin(\rho)$ of the ensemble average state~$\rho
\coloneqq \sum_x p_x \rho_x$. (Although the proof is written assuming
that the shared state consists of EPR pairs and some ancilla and an
auxiliary error parameter, it may be modified to give a bound when an
arbitrary state is shared and the auxiliary error is~$0$.) In this
article, we show that there are ensembles for which the min-entropy
bound equals the number of qubits in the states, and the bound
holds up to an additive constant even with the more general compression
protocols we allow.
\begin{theorem}
\label{thm-result}
There exist universal constants~$c_1 , c_2 >0$ such that for any~$\epsilon \in (0,1)$, and any~$k,m \in \naturals$ with~$k \geq 6/(1-\epsilon)$ and~$m \geq c_1 (\ln k)/(1-\epsilon)^2$ such that~$k$ divides~$m$, there exists an ensemble~$\left( \left( \tfrac{1}{n}, \rho_x \right) : x \in [n],   ~ \rho_x  
\in \qstate( \complex^m) \right)$, where~$n$ depends on~$k,m$, and~$\epsilon$, such that
\begin{itemize}
\item[(i)] 
$\rI(A:B)_{\tau} = \Imax(A:B)_{\tau} = \log_2 k$, where~$\tau \coloneqq \tfrac{1}{n} \sum_{x \in [n]} \density{x}^{A} \tensor \rho_x^{B}$ ~;
\item[(ii)] 
there is a one-way protocol with shared entanglement for the visible compression of the ensemble with average error~$\epsilon/2$ and with communication cost~$\tfrac{1}{2} \log k + \Order(\log \log \tfrac{1}{\epsilon})$; and
\item[(iii)] 
the sum of communication and entanglement costs of \emph{any\/} one-way protocol with shared entanglement for visible compression of the ensemble, with average-error at most~$\epsilon/2$, is at least
\[
\log m - 3 \log \frac{1}{1 - \epsilon } - c_2 \enspace.
\]
\end{itemize}
\end{theorem}
In particular, the theorem implies that in the absence of shared entanglement, the ensemble may only be
compressed by a constant number of qubits (independent of~$m$), even if constant average error~$\tfrac{\epsilon}{2} < 1/2$ is allowed. Note also that the straightforward protocol that prepares and sends the state~$\rho_x$ on input~$x$ has sum of entanglement and communication costs equal to~$\log m$. So the lower bound in the theorem is optimal up to an additive universal constant term for 
constant~$\epsilon \in (0,1)$.

Proposition~\ref{prop-ub} and Corollary~\ref{cor-ent} in Section~\ref{sec:Main-Result} 
contain more precise statements of the results stated in the theorem. As we explain in
that section, $\rI(A:B)_{\tau}$ may be interpreted as the ``information content'' of the ensemble; 
it is the \emph{quantum information cost\/}~\cite{Touchette15-QIC} of the
protocol in which Alice simply prepares the state~$\rho_x$ on input~$x$
and sends the state to Bob.

The compression task we study is a relaxation of oblivious (or
\emph{blind\/})
compression, in which the input to Alice is the state~$\rho_x$, rather
than~$x$. It is also a relaxation of state-splitting (more generally, of
\emph{state re-distribution\/}~\cite{LD09-QSR, DY08-redistribution,
YD09-quantum-coding}), and channel simulation. So the lower bound in
Theorem~\ref{thm-result}(ii) holds for these tasks as well.

The ensemble mentioned in Theorem~\ref{thm-result} is obtained via the
probabilistic method, and is of a form devised by Jain, Radhakrishnan,
and Sen~\cite{JRS03-direct-sum}.  They showed the incompressibility of 
such an ensemble when the decompression operation is unitary (i.e., via
protocols as in Figure~\ref{fig-compression} in which the register~$B_1$
is trivial). We adapt their proof method to protocols which allow a general 
quantum channel for decompression. A key step here is a technical lemma 
(Lemma~\ref{lemma-main} in Section~\ref{sec:Main-Result}) which allows
us to reason about general quantum channels, and also yields a tighter
lower bound on the sum of communication and entanglement costs.

\subsection{Implications and related work}

Jain \textit{et al.\/}~\cite{JRS05-compression, JRS08-direct-sum},
also used the same kind of ensemble as in Theorem~\ref{thm-result}
to design a two-party one-way communication 
protocol \emph{with shared entanglement\/} for the Equality function. They
showed that the initial shared state in the protocol cannot be replaced by
one with polynomially smaller dimension in a ``black-box fashion'' (i.e.,
when the local operations of the two parties are not modified).
Theorem~\ref{thm-result} implies a similar impossibility result for
protocols in which the sender and receiver can deviate from the original
protocol arbitrarily, but they try to approximate the receiver's state 
in the original protocol after the message is sent. The impossibility 
holds even when the dimension of the initial shared entangled state is 
reduced only by a constant factor.

A remarkable property of the ensemble posited by
Theorem~\ref{thm-result} is that the communication cost of compression
(with shared entanglement) may be arbitrarily smaller than the entanglement
cost. For constant error the communication cost is within an additive
constant of the quantum information cost~\cite{Touchette15-QIC} of the
protocol that simply prepares and sends the state.  As a consequence, we
infer that the quantum information cost of a protocol may be arbitrarily
smaller than the communication cost of any protocol \emph{without
shared entanglement\/} for compressing its messages.  Anshu, Touchette, Yao,
and Yu~\cite{ATYY17-separation} had previously proven a similar
separation when the compression protocol \emph{is\/} allowed to use
shared entanglement. However, their separation is exponential: they exhibited
an interactive protocol for a Boolean function with quantum information
cost that is exponentially smaller than the communication cost of any
interactive quantum protocol that computes the function. (Observe that a
protocol for compressing the final state of the original protocol may
also be used to compute the function.) In contrast to that protocol, the
one we present \emph{is\/} compressible to its quantum information cost,
but requires an arbitrarily larger amount of shared entanglement to do so.

In another related work, Liu, Perry, Zhu, Koh, and 
Aaronson~\cite{LPZKA16-separation} show that one-way protocols
cannot be compressed to their quantum information cost without using
shared entanglement. They consider a certain one-way protocol in
which Alice gets an~$n$-bit input, Bob gets an~$m$-bit input, with~$m
\in \order(n)$. The protocol has quantum information 
cost~$\Order(n m^{-2} \log m)$. They show that the protocol cannot be 
compressed by a one-way protocol without shared entanglement into a
message of length~$\order(\log n)$ with error at
most~$(n + 1)^{- m}$. Thus the separation is limited, and only holds 
for exponentially small error (in the length of the inputs).

It is believed that the communication in any 
interactive quantum protocol which has a constant number of rounds 
and computes a function of classical inputs may be compressed,
with constant error, to an amount proportional to the quantum 
information cost of the protocol. For one-way protocols such a result 
was shown by Jain, Radhakrishnan, and 
Sen~\cite{JRS05-compression,JRS08-direct-sum}. This was later re-proven by
Anshu, Jain, Mukhopadhyay, Shayeghi, and Yao~\cite{AJMSY16-one-shot} using
different techniques. A similar result for protocols with a larger constant
number of rounds of communication was claimed by 
Touchette~\cite{Touchette15-QIC}, but the proof has an error.
The compression protocols achieving quantum information cost all rely on the
presence of shared entanglement. Theorem~\ref{thm-result} shows that
even for the simplest protocols, such compression is not possible in 
the absence of shared entanglement. Moreover, it shows that the entanglement 
cost may be necessarily within an additive constant of the length of 
the message to be compressed, even when the quantum information cost 
is arbitrarily smaller than the message length.

In a recent independent work, Khanian and
Winter~\cite{KW19-pure-state-compression} analyse the communication and
entanglement costs of a variant of compression in the asymptotic
setting. They study pure state ensembles with quantum side information
in the form of pure states. In the case of visible compression with
shared entanglement, they show that the asymptotic (per-instance)
communication cost is at least~$\tfrac{1}{2} \, \rS(\rho)$, i.e., half
the entropy of the ensemble average state~$\rho$. So this cost may be at
most a factor of~$1/2$ smaller than that of compression \emph{without\/}
shared entanglement.  Moreover, the asymptotic sum of communication and
entanglement costs is at least the entropy~$\rS(\rho)$. Thus the kind of
separation we show does not hold for pure states even in the asymptotic
setting.

\paragraph{Organization.} The rest of this article is organized as follows.
In Section~\ref{sec:Preliminaries}, we review basic concepts and notation
from quantum information and communication.
In section~\ref{sec:Main-Result}, we prove the main result and discuss
its implications.

\paragraph{Acknowledgements.}
We thank Mil{\'a}n Mosonyi for extensive, thoughtful feedback on 
earlier versions of this article. This research is supported in part by NSERC Canada. SBH is also supported by an Ontario Graduate Scholarship.

\section{Preliminaries}
\label{sec:Preliminaries}

\subsection{Mathematical notation and background}

We refer the reader to the book Watrous~\cite{W18-TQI} for a thorough
introduction to basics of quantum information. We briefly review the
notation and some results that we use in the article. 

For the sake of brevity, we denote the set~$\{1,2,\ldots,k\}$ by~$[k]$.
We denote physical quantum systems (``registers'') with capital letters,
like~$X$,~$Y$ and~$Z$. The state space corresponding to a register is a 
finite-dimensional Hilbert space.
We denote (finite dimensional) Hilbert spaces either by capital script 
letters like~$\cH$ and $\cK$, or as~$\complex^m$ where~$m$ is the dimension. 
We denote the the dimension of a Hilbert space corresponding to a register~$X$ as~$|X|$.
We use the Dirac notation, i.e., ``ket'' and ``bra'', for unit vectors and
their adjoints, respectively. We denote 
the set of all unit vectors in a Hilbert space~$\cH$ by~$\sphere(\cH)$.
For a Hilbert space~$\cH \eqdef \complex^S$ for some non-empty
finite set~$S$, we call~$\set{ \ket{x} : x \in S}$ its \emph{canonical\/} basis. 

A subset~$N$ of~$\sphere(\cH)$ is called~$\epsilon$-\emph{dense\/} if for 
every vector~$\ket{u} \in \sphere(\cH)$, there exists a vector in the 
set~$N$ at Euclidean distance at most~$\epsilon$ from~$\ket{v}$. Such
a set is also called an ``$\epsilon$-net'' in the literature.
The following proposition states that every finite dimensional Hilbert 
space has a relatively small~$\epsilon$-dense set.

\begin{proposition}[\cite{Matousek02-discrete-geometry}, Lemma~13.1.1,
Chapter~13]
\label{Fact-epsilon-net}
Let~$\epsilon\in(0,1]$, and
$m$ be a positive integer. The Hilbert space~$\complex^{m}$ has 
an~$\epsilon$-dense
set~$N$ of size~$\abs{N}\leq\left(\frac{4}{\epsilon}\right)^{2m}$.
\end{proposition}
A slightly better bound~$\left(1+\tfrac{2}{\epsilon}\right)^{2m}$ on
the size of an~$\epsilon$-dense
set is given in Ref.~\cite[Lemma 2.6]{MS86-asymp-thy-normed-spaces}. 

We denote the set of all linear operators on Hilbert space~$\cH$
by~$\linear(\cH)$, the set of all positive semi-definite operators
by~$\Pos(\cH)$, the set of all unitary operators by~$\unitary(\cH)$, and
the set of all quantum states (or ``density operators'') over~$\cH$
by~$\qstate(\cH)$. The identity operator on~$\cH$ is denoted
by~$\id_\cH$.  We denote quantum states or sub-normalized states (positive
semi-definite operators with trace at most~$1$) by 
lowercase Greek letters like~$\rho, \sigma$. We use notation such 
as~$\rho^X$ to indicate that register~$X$ is in state~$\rho$, and 
may omit the superscript when the register is clear from the context.
An operator~$M \in \Pos(\cH)$ is called a \emph{measurement operator\/}
if~$M \leq \id$. We usually denote quantum channels, i.e., completely
positive trace-preserving linear maps from the space of linear operators
on a Hilbert space to another such space, by capital Greek letters
like~$\Psi$. The \emph{partial trace\/} over a Hilbert space~$\cK$
is denoted as~$\trace_{\cK}$.

We denote the operator norm (\textit{Schatten~$\infty$ norm\/}) of an
operator~$M \in \linear(\cH)$ by~$\norm{M}$, the Frobenius norm
(\textit{Schatten~$2$ norm\/}) by~$\norm{M}_\rF$, and the trace norm
(\textit{Schatten~$1$ norm}) by~$\trnorm{M}$. Recall that~$\trnorm{M}
\coloneqq \trace \sqrt{M^* M}$ is the sum of the singular values of~$M$,
$\norm{M}$ is the largest singular value, and~$\norm{M}_\rF \coloneqq
\sqrt{ \trace (M^* M) }$ is
the~$\ell_2$-norm of the singular values with multiplicity.
All of these norms are invariant under
composition with a unitary operator. 

We consider random unitary operators chosen according to the
\emph{Haar measure\/}~$\eta$ on~$\unitary(\cH)$, where~$\cH$ is a 
finite dimensional Hilbert space. The
Haar measure is the unique unitarily invariant probability 
measure over~$\unitary(\cH)$. 

Let~$f:\unitary(\cH)\rightarrow \reals$ be a continuous function. 
Suppose~$f$ is $\kappa$-Lipschitz, i.e., for all~$U,V \in \unitary(\cH)$,
we have
\[
\size{ f(U) - f(V) } \quad \le \quad \kappa \norm{ U - V}_\rF \enspace,
\]
for some~$\kappa \ge 0$. If~$\kappa$ is small enough as compared to the
dimension of~$\cH$, with high probability, the random variable~$f(\bU)$
is close to its expectation, where~$\bU \in \unitary(\cH)$ is a
Haar-random unitary operator.  This \emph{concentration of measure\/}
property is formalized by the following theorem, which is a special case
of Theorem~5.17 in Ref.~\cite{M19-random-matrix-theory}.
\begin{theorem}[\cite{M19-random-matrix-theory}, Theorem~5.17, page~159]
\label{lemma-concentration} 
Let~$\eta$ be the Haar measure on~$\unitary(\cH)$, where~$\cH$ is a
Hilbert space with finite dimension~$m$, and let~$\bU \in \unitary(\cH)$ 
be a random unitary operator chosen according to~$\eta$.
For every function~$f : \unitary(\cH)\rightarrow\reals$ that
is~$\kappa$-Lipschitz with respect to the Frobenius norm (with~$\kappa > 0$),
and every positive real number~$t$, we have
\[
\eta \Big( \{ U\in\unitary(\cH) : f(U)-\expct \left[f(\bU)\right]\geq t\}
        \Big)
    \quad \leq \quad \exp \!\left( - \frac{ (m-2) t^{2} }{ 24 \kappa^2 }
        \right) \enspace.
\]
\end{theorem}

The \emph{fidelity} between two sub-normalized states~$\rho$ and~$\sigma$ is
defined as
\[
\rF(\rho,\sigma) \quad \coloneqq \quad
    \trace \sqrt{\sqrt{\rho}\ \sigma\sqrt{\rho}}
    + \sqrt{(1 - \trace(\rho)) (1 - \trace(\sigma)) } \enspace.
\]
Fidelity can be used to define a useful metric called the \emph{purified 
distance\/}~\cite{GLN05-dist-meas,TCR10-duality} between sub-normalized states:
\[
\rP(\rho,\sigma)\quad\coloneqq\quad\sqrt{1-\rF(\rho,\sigma)^{2}}\enspace.
\]
For a quantum state~$\rho \in \qstate(\cH)$ and~$\epsilon\in[0,1]$, we
define
\[
\sB^{\epsilon}(\rho) \quad \coloneqq \quad
    \set{ \widetilde{\rho} \in \Pos(\cH) : \rP(\rho, \widetilde{\rho}) 
    \leq \epsilon, ~ \trace{\:\widetilde{\rho}} \leq 1 }
\]
as the ball of sub-normalized states that are within purified 
distance~$\epsilon$ of~$\rho$.

The trace distance between quantum states is induced by the trace 
norm.  The following property is well known (see, e.g., 
Ref.~\cite[Theorem~3.4, page~128]{W18-TQI}).
\begin{proposition}[Holevo-Helstrom Theorem \cite{Hel67-detection,Hol72-decision}]
\label{prop-trnorm}
For any pair of quantum states~$\rho, \sigma \in \qstate(\cH)$,
\[
\trnorm{\rho - \sigma} \quad = \quad 2 \; \max \;
    \set{ \; \abs{ \trace(M \rho) - \trace(M \sigma) } :  M \textrm{ is a
    measurement operator on } \cH } \enspace.
\]
\end{proposition}
Purified distance and trace distance are related to each other as
follows (see, e.g., Ref.~\cite[Theorem~3.33, page~161]{W18-TQI}):
\begin{proposition}[Fuchs and van de Graaf Inequalities \cite{FG99-dist-meas}]
\label{prop-FvdG}
For any pair of quantum states~$\rho, \sigma \in \qstate(\cH)$,
\[
1 - \sqrt{1- \rP(\rho,\sigma)^2}
    \quad \leq \quad \frac{1}{2} \trnorm{\rho - \sigma}
    \quad \leq \quad \rP(\rho, \sigma) \enspace.
\]
\end{proposition}


Unless specified, we take the base of the logarithm function to be~$2$.

Let~$\cH, \cK$, and~$\cM$ be the state spaces corresponding to 
registers~$X, Y$, and~$M$, respectively. For a register~$X$ in
quantum state~$\rho \in \qstate(\cH)$, the
\emph{von Neumann entropy\/} of~$X$ is defined as
\[
\rS(\rho) \quad \coloneqq \quad -\trace \left( \rho \log \rho \right) \enspace.
\]
This coincides with the Shannon entropy of the spectrum of~$\rho$.
The \emph{relative entropy\/} of two quantum states $\rho, \sigma
\in \qstate(\cH)$ is defined as
\[
\rS(\rho \| \sigma) \quad \coloneqq \quad \trace \left( \rho \log \rho - \rho 
    \log \sigma \right) \enspace,
\]
when~$\support(\rho) \subseteq \support(\sigma)$, and is~$\infty$
otherwise. The
\emph{max-relative entropy\/} of~$\rho$ with respect to~$\sigma$ is 
defined as
\[
\Smax( \rho \| \sigma) \quad \coloneqq \quad
    \min \{ \lambda : \rho \leq 2^{\lambda} \sigma \} \enspace,
\]
when~$\support(\rho) \subseteq \support(\sigma)$, and is~$\infty$
otherwise. The \emph{min-entropy\/} of~$\rho$ is defined as
\[
\Smin( \rho) \quad \coloneqq \quad - \log \norm{ \rho} \enspace.
\]

Suppose that the registers~$X,Y$ are in joint
state~$\rho^{XY} \in \qstate(\cH \otimes \cK)$.
The \emph{mutual information\/} of~$X$ and~$Y$ is defined as
\[
\rI(X:Y)_{\rho} \quad \coloneqq \quad \rS \!\left( \rho^{XY} \|
    \rho^{X} \tensor \rho^{Y} \right) \enspace.
\]
When the state is clear from the context, the subscript~$\rho$ may be
omitted from the notation. When~$\rho$ is a classical-quantum state,
i.e.,~$\rho^{XY} = \sum_{x} p_x \density{x}^X \tensor
\rho_x^Y$ with~$p$ being a probability distribution, $\set{ \ket{x}}$ 
the canonical orthonormal basis for~$\cH$, and~$\rho_x\in\qstate(\cK)$,
we have
\[
\rI(X:Y) \quad = \quad \sum_x p_x \; \rS(\rho_x \| \rho)\enspace,
\]
where~$\rho=\sum_x p_x \rho_x$.
Suppose the registers~$X,Y,M$ are in joint (tripartite) state~$\rho^{XY
\! M}
\in \qstate(\cH \tensor \cK \tensor \cM)$. The \emph{conditional
mutual information\/} of~$X$ and~$M$ given~$Y$ is defined as
\[
\rI(X:M \,|\, Y) \quad \coloneqq \quad \rI(XY:M) - \rI(Y:M) \enspace.
\]
When~$\rho^{XY \! M}$ is a tensor product of the states~$\rho^{X \! M}$
and~$\rho^Y$, we have
\[
\rI(X:M \,|\, Y) \quad = \quad \rI(XY:M) \quad = \quad \rI(X:M) \enspace.
\]

For any state~$\rho^{XY} \in \qstate(\cH \otimes \cK)$, the
\textit{max-information} register~$Y$ has about
register~$X$~\cite{BCR11-reverse-Shannon} is defined as
\[
\Imax(X:Y)_\rho \quad \coloneqq \quad
    \min_{\sigma \in \qstate(\cK)} \Smax \!\left( \rho^{XY} \,\|\,
    \rho^{X} \otimes \sigma^Y \right) \enspace.
\]
For a parameter~$\epsilon \in [0,1]$, the \emph{smooth max-information\/} 
register~$Y$ has about register~$X$ is defined as
\[
\Imax^{\epsilon}(X:Y)_{\rho} \quad \coloneqq \quad
    \min_{\widetilde{\rho} \in \sB^{\epsilon}(\rho)}
    \Imax(X:Y)_{\widetilde{\rho}} \enspace.
\]

\subsection{Quantum communication protocols}
\label{sec-qcp}

We first describe a two-party quantum communication protocol informally
and then give a formal definition for the special case of interest to us. 
We refer the reader to, e.g., Ref.~\cite{Touchette15-QIC} for a formal
definition of the general case.

In a two-party quantum communication protocol, there are two parties,
Alice and Bob, each of whom may get some input in registers designated
for this purpose. Alice and Bob's inputs may be entangled with each
other, and also with a ``reference'' system, which purifies it. 
Alice and Bob's goal
is to accomplish an information processing task by communicating with
each other.

Each party possesses some ``work'' (or ``private'') qubits (or
registers) in addition to the input registers. The work qubits are
initialized to a fixed pure state in tensor product with the input
state.  This fixed state may be entangled across the work registers of Alice
and Bob, and may be used as a computational resource. In this case, we
say the protocol or the channel is \emph{with shared entanglement\/} or
with \emph{entanglement assistance\/}.  If the fixed state is a tensor
product state across Alice and Bob's registers, we say it is a protocol
or channel \emph{without shared entanglement\/} or simply 
\emph{unassisted\/}.

The protocol proceeds in some number of ``rounds''.
In each round, the sender applies an isometry to the qubits in her
possession, and sends a sub-register (the message) to the other party.
The length of the message (in qubits) is the base~$2$ logarithm of the
dimension of the message register.
After the last round, the recipient of the last message applies an
isometry to his registers. The output of the protocol is the
state of a pair of designated registers of the two parties at the end.

We are often interested in minimizing the total length of the messages
over all the rounds, i.e., the \emph{communication cost\/} (or
\emph{complexity\/}) of the protocol.
The idea is to accomplish the task at hand with minimum communication.
In protocols with shared entanglement, we are also interested in the
amount of shared entanglement needed in the protocol, i.e., the minimum
dimension of the support of the \emph{initial\/} state of either 
party's work space. This latter quantity, measured in number of qubits, 
is called the \emph{entanglement cost\/} of the protocol.

In this article, we study only \emph{one-way\/} protocols, 
i.e., protocols with one round, and therefore one message, (say) from Alice 
to Bob. We describe these more formally here.
Alice and Bob initially hold registers~$A_\rin E_\rA$ and~$B_\rin
E_\rB$, respectively. The input registers~$A_\rin B_\rin$ are
initialized to some state~$\rho^{A_\rin B_\rin}$ whose purification is
held in register~$R$ with a third party, the referee. 
Alice and Bob's work registers~$E_{\rA}$ and~$E_{\rB}$ are initialized 
to a pure state~$\ket{\phi}^{E_{\rA} E_{\rB}}$, which may be entangled
across the partition~$E_\rA E_\rB$.
The local operations in the protocol are specified by two 
isometries~$U$ and~$V$.
The isometry~$U$ acts on registers~$A_{\rin} E_{\rA}$ and maps them to
registers~$A_{\out} A_{1} M$. The isometry~$V$ acts on 
registers~$B_{\rin} E_{\rB} M$ and maps them to registers~$B_1 B_{\out}$. 
First, Alice applies~$U$ to the registers~$A_{\rin}$ and~$E_{\rA}$ and 
sends the register~$M$ to Bob. Then, Bob applies~$V$ on his initial 
registers~$B_{\rin} E_{\rB}$ and the message~$M$. The output of the 
protocol is the state of Alice and Bob's registers~$A_\out B_\out$.
The communication cost of this protocol is~$\log |M|$ and the entanglement 
cost is the logarithm of the Schmidt rank of the state~$\ket{\phi}$
across the partition~$E_\rA E_\rB$. We say it is a protocol \emph{with
shared entanglement\/} if the Schmidt rank of~$\ket{\phi}$ is more
than~$1$, and say that it is \emph{without shared entanglement\/} otherwise.
Such protocols are also called \emph{entanglement-assisted\/} and
\emph{unassisted\/}, respectively, in the literature. 

We say that the input is ``classical'' when there are non-empty finite 
sets~$S_\rA, S_\rB$ (the sets of classical inputs) such that the Hilbert 
spaces corresponding to the input registers are~$\complex^{S_\rA}, 
\complex^{S_\rB}$, respectively, and the initial joint quantum state 
in the input registers~$A_\rin B_\rin$ is diagonal in the 
canonical basis~$\set{ \ket{x} \ket{y} : x \in S_\rA, y \in S_\rB}$.
In the case that the inputs to Alice and Bob are classical, we assume 
without loss of generality that the input registers~$A_\rin$
and~$B_\rin$ are ``read-only'', i.e., the isometries~$U$ and~$V$ 
are of the form~$\sum_{x \in S_\rA} \density{x}^{A_\rin} \tensor
U_x^{E_\rA}$ and~$\sum_{y \in S_\rB} \density{y}^{B_\rin} \tensor
V_y^{M E_\rB}$, where~$S_\rA, S_\rB$ are sets as above.
A one-way protocol in which Alice gets a classical input and 
Bob does not have any input is depicted in Figure~\ref{fig-compression}.

Let~$\Pi$ be a one-way quantum protocol (with or without shared entanglement)
with a single message from Alice to Bob, in which Alice gets a 
classical input and Bob does not have any input. The register~$R$ with
the referee purifies Alice's input so that~$\ket{\rho}^{R A_\rin}
\coloneqq \sum_{x \in S_\rA} \sqrt{p_x} \ket{xx}^{R A_\rin}$,
where~$p_x$ is a probability distribution over the input set~$S_\rA$.
Let~$M$ be the quantum register corresponding to the message in~$\Pi$.
The \emph{quantum information cost\/} (or \emph{quantum information 
complexity\/}) of the protocol~$\Pi$ is defined as
\[
\QIC(\Pi) \quad \coloneqq \quad \frac{1}{2} \; \rI( R : M \,|\, E_{\rB})
\enspace,
\] 
where the registers are in the state immediately after Alice sends the
message register~$M$ to Bob. This expression simplifies to~$\rI( R : M
E_{\rB})$ as the registers~$R, E_{\rB}$ are in a tensor product state at this
point. It is intended to measure the information Bob gains about Alice's
input from the message.  This notion requires a nuanced definition
for protocols with more general inputs and with multiple rounds of
communication. As it is not central to our work, we refer the reader to
Ref.~\cite{Touchette15-QIC} for the definition for general protocols.

\subsection{Compression of quantum states}
\label{sec-compression}

We study one-way protocols for \emph{non-oblivious\/} or
\emph{visible\/} compression of quantum states, which is typical for
tasks of this nature (see, e.g.,
Ref.~\cite{AJ18-efficient-convex-split}). The protocol may be with or
without shared entanglement. Suppose we wish to compress states chosen
from an ensemble~$( (p_x, \rho_x) : x \in S)$ for some finite set~$S$,
where~$p$ is a probability distribution over~$S$ and~$\rho_x \in
\qstate(\cH)$. The ensemble is known to both parties.  The sender, say
Alice, is given a classical input~$x \in S$ chosen according to the
distribution~$p$. Alice and Bob execute a one-way protocol with a
message from Alice to Bob in order to prepare an approximation
of~$\rho_x$ on Bob's side. Following the notation from
Section~\ref{sec-qcp}, we interpret the state of the message
register~$M$ of this protocol as a compression of~$\rho_x$. Suppose the
state of the output register~$B_\out$ is~$\widetilde{\rho}_x$.
We say that the average error of the compression protocol is~$\epsilon
\in [0,2]$ if the output state~$\widetilde{\rho}_x$ is~$\epsilon$-close
in trace distance to the ideal state~$\rho_x$ on average over the
inputs~$x$:
\[
\sum_{x} p_x \trnorm{ \rho_{x} - \widetilde{\rho}_{x} }
    \quad \le \quad \epsilon \enspace.
\]
It is sometimes desirable to express the error in terms of the
purified distance. For simplicity, we state error bounds in terms of trace 
distance; we may express the bounds in terms of purified distance 
via Proposition~\ref{prop-FvdG}.

Note that a protocol for visible compression without shared entanglement
may be characterized by a sequence of quantum states~$(\sigma_{x}: x \in
S)$ and a quantum channel~$\Psi$. We let~$\sigma_x$ be the state of the
message register~$M$ sent by Alice to Bob on input~$x$. We define~$\Psi$
as the channel resulting from the application of the isometry~$V$
followed by the tracing out of the register~$B_1$.  The average error of
the protocol is then~$\sum_{x} p_x \trnorm{ \rho_{x} - \Psi(\sigma_{x})
}$.  Conversely, any choice of states~$(\sigma_x : x \in S, ~ \sigma_x
\in \qstate( \cK))$ and quantum channel~$\Psi : \linear( \cK) \rightarrow
\linear( \cH)$ for some Hilbert space~$\cK$ defines a valid visible 
compression protocol.

An essentially equivalent formulation of the task of visible compression
is the following (with the notation from Section~\ref{sec-qcp}).
Consider the state~$\tau$ over the registers~$R X A_1 C$:
\[
\tau \quad \coloneqq \quad \sum_{x \in S} \sqrt{p_x} \ket{x}^R
    \ket{x}^X \ket{\phi_x}^{A_1 C} \enspace,
\]
where~$\ket{\phi_x}^{A_1 C}$ is a purification of~$\rho_x$, register~$R$
is held by the referee, and registers~$X A_1 C$ together constitute
Alice's input register~$A_\rin$. Alice and Bob both know the full
description of~$\tau$. Their goal is to run a one-way quantum
communication protocol with a message from Alice to Bob, with or without
shared entanglement, such that at the end, the state~$\widetilde{\tau}$
of registers~$R B_\out$ is close to~$\tau^{R C}$: 
\[
\trnorm{ \widetilde{\tau}^{R B_\out} - \tau^{R C} } \quad \le \quad
\epsilon \enspace.
\]
The difference from state-splitting is that for a fixed state~$\ket{x}$
of register~$R$, the purification of the state in register~$B_\out$ may be
shared arbitrarily between Alice and Bob (while in state splitting, it
is required to be held by Alice, in register~$A_1$). A protocol for
state-splitting can thus be used for this task, and conversely lower bounds
on communication or entanglement costs derived for the above task
applies to state-splitting as well.

\section{The main result}
\label{sec:Main-Result}

In this section, we prove the main result of this article. 

\subsection{Two useful lemmas}

We begin with two lemmas that we need for the result.
The first allows us to focus on a finite number of
subspaces of a finite dimensional Hilbert space, in the context 
of measurements. For an operator~$M \in \linear(\cH)$, and a subspace~$\cA$
of~$\cH$, define the semi-norm 
\[
\norm{M}_\cA \quad \coloneqq \quad
    \max_{\ket{w} \; \in \; \sphere(\cA)}
    \size{ \bra{w} M \ket{w} } \enspace.
\]

\begin{lemma}[\cite{JRS03-direct-sum}, Lemma~6]
\label{lemma-net-W} 
Let~$d$ and~$q$ be positive integers with~$q \ge d$,
$\delta>0$ be a real number, and~$\cH$ be an~$q$-dimensional Hilbert
space. There exists a set~$\fT$ of subspaces of~$\cH$ of dimension at
most~$d$ such that
\begin{enumerate}

\item $\size{\fT} \leq \left( \frac{8 \sqrt{d}}{\delta} \right)^{2 q d}$,
and

\item for every~$d$-dimensional subspace~$\cA \subseteq\cH$, there
is a subspace~$\cB \in \fT$ such that for every measurement operator~$M
\in \Pos(\cH)$,
\[
\Big| \norm{M}_\cA - \norm{M}_\cB \Big| \quad \leq \quad \delta \enspace.
\]
\end{enumerate}
\end{lemma}
The set~$\fT$ in the lemma is obtained as follows. We fix an~$\epsilon$-dense 
subset~$S$ of~$\sphere(\cH)$ for a suitably small value of~$\epsilon$,
as given by Proposition~\ref{Fact-epsilon-net}.
For any $d$-dimensional subspace~$\cA$, we consider an orthonormal basis, and
the~$d$ vectors in~$S$ closest to the respective elements in the basis.
We include in~$\fT$ the subspace~$\cB$ spanned by the~$d$ vectors from~$S$ so
obtained. 

By a uniformly random subspace of dimension~$\ell$ of an~$m$-dimensional 
Hilbert space~$\cH$, with~$\ell \le m$, we mean the image of a
fixed~$\ell$-dimensional subspace under a Haar-random unitary operator 
on~$\cH$. The next lemma is similar to Lemma~7 from
Ref.~\cite{JRS03-direct-sum}, and is stronger in several respects. It
enables the generalization of the incompressibility result in
Ref.~\cite{JRS03-direct-sum} that we prove, and helps us derive tighter 
bounds for compression. Informally, the lemma states that every state in a
``small enough'' subspace of a bi-partite space has, with high probability,
a small projection onto a ``small enough'' random subspace of one part.
\begin{lemma}
\label{lemma-main} 
Let~$m$, $d$, $\ell$, and~$p$ be positive integers such that~$\ell \le m$. 
Let~$\cW$ be a 
fixed~$d$-dimensional subspace of~$\complex^m \tensor \complex^p$.
Let~$\bcZ$ be a uniformly random subspace of~$\complex^{m}$ of 
dimension~$\ell$, and~$\bM$ be the orthogonal projection operator 
onto~$\bcZ$. Then for any real number~$\alpha > 2$, there is a real
number~$\alpha_1 > 0$ that depends only on~$\alpha$ such that
\[
\Pr\left[ \norm{\bM \tensor \id_{\complex^p}}_{\cW} \geq
        \frac{\alpha \ell}{m} \right]
    \quad \leq \quad \exp\left(-\frac{ \alpha_1 \ell^2 (m-2) }{m^2} \right)
        \enspace,
\]
provided
\[
(\alpha - 2)^2 \ell^2 (m - 2) \quad \ge \quad (4 \times 384) d m^2
    \ln \left( \frac{8 m}{\alpha \ell} \right) \enspace.
\]
We may take~$\alpha_1 \coloneqq \tfrac{(\alpha-2)^{2}}{768}$ in the above
statement.
\end{lemma}

\begin{proof}
The subspace~$\cW$ is isomorphic to~$\complex^d$ as it is~$d$-dimensional.
By Proposition~\ref{Fact-epsilon-net}, there is a set~$N$
with~$\size{N} \leq \left( \frac{8 m}{\alpha \ell} \right)^{2d}$
that is a~$\frac{\alpha \ell}{2 m}$-dense set of~$\sphere(\cW)$.

Note that for any two vectors~$\ket{u},\ket{v} \in \sphere(\complex^{m} \tensor 
\complex^{p})$, we have
\begin{IEEEeqnarray*}{rCl}
\abs{ \bra{u} (\bM \tensor \id) \ket{u} - \bra{v} (\bM \tensor \id) \ket{v}} \quad
    & = & \quad \size{ \trace \left( \bM \density{u} - \bM \density{v}
        \right) } \\
    & \leq & \quad \frac{1}{2} \trnorm{\density{u} - \density{v} } \qquad \qquad 
        \textrm{(by Proposition~\ref{prop-trnorm})} \\
    & \leq & \quad \frac{1}{2} \trnorm{(\ket{u}-\ket{v})\bra{v}} 
        + \frac{1}{2}\trnorm{\ket{u}(\bra{u}-\bra{v})} \\
    & = & \quad \norm{\ket{u}} \norm{\ket{u}-\ket{v}} \quad = \quad  
        \norm{\ket{u}-\ket{v}} \enspace.
\end{IEEEeqnarray*}
This implies that if~$\norm{ M \tensor \id_{\complex^p} }_\cW \geq
\tfrac{\alpha \ell}{m}$, there
is a vector~$\ket{v} \in N$ such that~$\bra{v} (\bM \tensor \id) \ket{v} \geq
\tfrac{\alpha \ell}{2 m}$. By the Union Bound, we get
\begin{IEEEeqnarray}{rCl}
\label{eq-bound1}
\Pr \left[ \norm{\bM \tensor \id_{\complex^p}}_{\cW} \geq
        \frac{\alpha \ell}{m} \right] \quad
    & \leq & \quad \size{N} \times \max_{\ket{v} \in N} 
        \Pr \left[ \bra{v} (\bM \tensor \id) \ket{v} \geq 
        \frac{\alpha \ell}{2 m} \right] \enspace.
\end{IEEEeqnarray}

Consider any fixed vector~$\ket{v} \in N$ and let~$P \in
\Pos(\complex^{m})$ be a fixed orthogonal projection of 
rank~$\ell$. Consider the 
function~$f : \unitary(\complex^{m}) \rightarrow \reals$ defined as
\[
f(U) \quad \coloneqq \quad \bra{v} \left( U P U^{*}
    \tensor \id_{\complex^p} \right) \ket{v} \enspace.
\] 
For any~$U,W \in \unitary(\complex^{m})$, we have
\begin{IEEEeqnarray*}{rCl}
\size{f(U) - f(W)} \quad
    & = & \quad \Big| \trace \big[ \left( \left( U P U^* - W P W^* \right)
        \tensor \id \right) \density{v} \; \big] \Big| \\
    & \le & \quad \norm{ U P U^* - W P W^* } \\
    & \le & \quad \norm{ U P U^* - W P U^* } + \norm{ W P U^* - W P W^* } \\
    & \le & \quad \norm{ U - W } + \norm{ U^* - W^* } \\
    & \le & \quad 2 \norm{U - W}_\rF \enspace.
\end{IEEEeqnarray*}
\suppress{
where~$\norm{ \cdot }$ denotes the spectral norm, and~$\norm{ \cdot }_\rF$ 
denotes the Frobenius norm.
}
So~$f$ is~$2$-Lipschitz.

Let~$\bU \in \unitary(\complex^{m})$ be a Haar-random unitary operation.
The expectation of~$f(\bU)$ is:
\begin{IEEEeqnarray*}{rCl}
\expct [ f(\bU) ] \quad
    & = & \quad \bra{v} \Big( \expct[ \bU P \bU^{*} ]
        \tensor \id \Big) \ket{v} \\
    & = & \quad \bra{v} \left( \ell
        \frac{\id}{m} \tensor \id \right) \ket{v} \\
    & = & \quad \frac{ \ell }{ m } \enspace.
\end{IEEEeqnarray*}
Since~$\bU P \bU^{*}$ and~$\bM$ have the same distribution, by
Theorem~\ref{lemma-concentration} we get
\begin{align*}
\Pr \left[ \bra{v} \left( \bM \tensor \id \right) 
        \ket{v} \geq \frac{\alpha \ell}{2 m} \right]
    \quad & = \quad \Pr \left[ \bra{v}
        \left( \bU P \bU^{*} \tensor \id \right) \ket{v} 
        \geq \frac{\alpha \ell}{2 m} \right] \\
    & \leq \quad \exp \!\left( -\frac{(m-2)(\alpha - 2)^{2} \ell^2}
        {384 m^2} \right) \enspace.
\end{align*}

By Eq.~(\ref{eq-bound1}), we get 
\begin{IEEEeqnarray*}{rCl}
\Pr \left[ \norm{\bM \tensor \id_{\complex^p}}_{\cW} \geq \frac{\alpha
        \ell}{m} \right] \quad
    & \leq & \quad \left( \frac{8 m}{\alpha \ell} \right)^{2d}
        \exp \left( -\frac{(m-2) (\alpha - 2)^{2} \ell^2 }{ 384 m^2 }
        \right) \\
    & \leq & \quad \exp \left( -\frac{(m-2) (\alpha - 2)^2 \ell^2 }{ 768 m^2 } 
        \right) \enspace,
\end{IEEEeqnarray*}
provided the~$m, \ell, d, \alpha$ satisfy the stated condition.
\end{proof}

\subsection{The ensemble and its compressibility}

We study an ensemble of the same form as in
Ref.~\cite{JRS03-direct-sum}. For positive integers~$n,m,k$ such
that~$k$ divides~$m$ and~$n$,
let~$B_{i} = \left( \ket{b_{i1}}, \ket{b_{i2}}, \dotsc, \ket{b_{im}} \right)$
be a suitably chosen orthonormal basis for~$\complex^{m}$, for each~$i
\in \left[ \tfrac{n}{k} \right]$. Let~$\left( B_{ij} : j \in [ k ]
\right)$ be a partition of~$B_i$ into~$k$ equal
size sets. Define~$\rho_{ij} \coloneqq \tfrac{k}{m} \sum_{\ket{v} \in B_{ij}} 
\ketbra{v}{v}$.
We show that there is a choice of bases such that the ensemble
\begin{equation}
\label{eq-ensemble}
\left( \left( \tfrac{1}{n}, \rho_{ij} \right) : i 
\in \left[ \frac{n}{k} \right], j \in [ k ]  \right)
\end{equation}
cannot be compressed significantly in the absence of shared entanglement. The
following theorem, which we prove along the same lines as Theorem~5
in Ref.~\cite{JRS03-direct-sum}, contains the crux of the argument.

\begin{theorem} 
\label{thm-compression}
Let~$\beta \in (0,1), ~ \epsilon \in ( 0, 2)$, and~$\nu \in (0, 1 - \epsilon/2 )$.
Let~$k, m, n, d$ be positive integers such that $k$ divides~$m$ and~$n$.
There exists an ensemble of~$n$ quantum states~$( \rho_{ij} )$ of the form
in~Eq.~(\ref{eq-ensemble}) such that for any sequence of quantum states~$\big(
\sigma_{ij} : \sigma_{ij} \in \qstate(\complex^{d}), \, i \in \left[
\tfrac{n}{k} \right], \, j \in [k] \big)$, and for all quantum 
channels~$\Psi : \linear( \complex^d ) \rightarrow \linear( \complex^m )$,
we have
\[
\Big| \set{ (i,j) : \trnorm{ \rho_{ij} - \Psi( \sigma_{ij} ) } > \epsilon } 
    \Big| \quad > \quad \beta n \enspace,
\]
when
\begin{IEEEeqnarray*}{rCl}
k \quad & \ge &\quad \frac{4}{1 - \epsilon/2 - \nu} \enspace, \\
m \quad & >   & \quad \max \left\{ \frac{3}{\gamma} \ln \!
    \left( \frac{\e}{ 1 - \beta} \right) , ~~
    \frac{3}{\gamma} \ln k , ~~
    2 + \frac{d}{ \gamma }
    \ln \! \left( \frac{16}{1 - \epsilon/2 - \nu} \right) \right\}
    \enspace, \qquad \textrm{and} \\
n \quad & > & \quad \frac{6 k d^2 m}{ \gamma (1 - \beta)} \;
    \ln \! \left( \frac{8 \sqrt{d} }{ \nu } \right) \enspace,
\end{IEEEeqnarray*}
where~$\gamma \coloneqq \tfrac{(1 - \epsilon/2 - \nu)^2}{ 8 \times 768}$.
\end{theorem}
 
\begin{proof}
We use the Probabilistic Method to show the existence of an ensemble
with the claimed property. We first derive a simpler property that
suffices.

For~$ i \in \left[ \tfrac{ n }{ k } \right]$ and~$j \in [ k ]$, 
let~$\tau_{ij} \in \qstate(\complex^{m})$ be~$m$-dimensional quantum
states and~$M_{ij}$ be the orthogonal projection onto the support of~$\tau_{ij}$.
By Proposition~\ref{prop-trnorm}, the condition
\begin{IEEEeqnarray}{rCl}
\label{eq-cond1}
\Big| \trace \left( M_{ij} \tau_{ij} \right)
        - \trace \left( M_{ij} \Psi( \sigma_{ij} ) \right) \Big| \quad
    & > & \quad \frac{ \epsilon}{ 2}
\end{IEEEeqnarray}
implies that~$\trnorm{ \tau_{ij} - \Psi( \sigma_{ij})} > \epsilon $.
Since~$\trace( M_{ij} \tau_{ij} ) = 1$, Eq.~(\ref{eq-cond1}) is
equivalent to
\begin{IEEEeqnarray}{rCl}
\label{eq-cond2}
\trace \left( M_{ij} \Psi( \sigma_{ij} ) \right) \quad
    & < & \quad 1 - \frac{ \epsilon}{ 2} \enspace.
\end{IEEEeqnarray}

Consider the following Stinespring representation~\cite[Corollary~2.27,
Sec.~2.2]{W18-TQI} of the quantum 
channel~$\Psi : \linear( \complex^d ) \rightarrow \linear( \complex^m )$ in
terms of a unitary operation~$ U \in \unitary( \cA \tensor \cB \tensor
\cC)$ and a fixed pure state~$\ancilla \in \cB \tensor \cC$, 
with~$\cA = \complex^d, \cB = \cC = \complex^m $:
\[
\Psi(\omega) \quad = \quad \trace_{ \cA \tensor \cB}
\Big[ U (\omega \tensor \density{ \bar{0} }) U^{*} \Big]
    \qquad\qquad  \forall \omega \in \linear( \complex^d ) \enspace.
\]
So we have
\begin{align*}
\trace \left( M_{ij} \Psi( \sigma_{ij} ) \right)
    \quad & = \quad \trace \Big( M_{ij} \; \trace_{ \cA \tensor \cB}
        \big[ U ( \sigma_{ij} \tensor \density{ \bar{0}} ) U^* \big]
        \Big) \\
    & = \quad \trace \Big( \left( M_{ij} \tensor \id_{ \cA \tensor \cB} 
        \right) U ( \sigma_{ij} \tensor \density{ \bar{0}} ) U^* \Big)
        \enspace,
\end{align*}
and Eq.~(\ref{eq-cond2}) is equivalent to
\begin{IEEEeqnarray}{rCl}
\label{eq-cond3}
\trace \Big( \left( M_{ij} \tensor \id_{ \cA \tensor \cB}
    \right) U ( \sigma_{ij} \tensor \density{ \bar{0}} ) U^* \Big) \quad
     & < & \quad 1 - \frac{ \epsilon}{ 2} \enspace.
\end{IEEEeqnarray}

For a fixed unitary operator~$U$, for any~$i,j$, 
the state~$U ( \sigma_{ij} \tensor \density{ \bar{0}} ) U^*$ 
belongs to~$\qstate(\cX)$ where~$\cX \coloneqq U( \cA \tensor \ancilla)$
is a fixed~$d$-dimensional subspace of~$\cA \tensor \cB \tensor \cC$.
Thus, the expression on the left in Eq.~(\ref{eq-cond3}) is bounded
by~$\norm{ M_{ij} \tensor \id_{ \cA \tensor \cB} }_\cX $ for every~$i,j$.
So it suffices to exhibit an ensemble such
that for all~$d$-dimensional subspaces~$\cW \subseteq \cA \tensor \cB
\tensor \cC$,
\begin{IEEEeqnarray*}{rCl}
\Big| \set{ (i,j) : \norm{ M_{ij} \tensor \id_{ \cA \tensor \cB} }_\cW <
        1 - \frac{ \epsilon}{ 2} } \Big| \quad
    & > & \quad \beta n \enspace.
\end{IEEEeqnarray*}
By Lemma~\ref{lemma-net-W}, for any~$\nu > 0$, 
there is a collection~$\fT$ of subspaces
of~$\cA \tensor \cB \tensor \cC$ of dimension at most~$d$, 
such that size~$\size{ \fT } \le ( 8
\sqrt{d} / \nu)^{ 2 d^2 m^2}$, and for all subspaces~$\cW$ as
above, there is a subspace~$\cY \in \fT$ such that for all~$i,j$,
\begin{IEEEeqnarray*}{rCl}
\Big| \norm{ M_{ij} \tensor \id_{ \cA \tensor \cB} }_\cW 
        - \norm{ M_{ij} \tensor \id_{ \cA \tensor \cB} }_\cY \Big| \quad
    & \le & \quad \nu \enspace.
\end{IEEEeqnarray*}
Taking~$\nu < 1 - \frac{ \epsilon}{ 2}$, it suffices to produce an 
ensemble such that for all subspaces~$\cY \in \fT$, 
\begin{IEEEeqnarray}{rCl}
\label{eq-cond4}
\Big| \set{ (i,j) : \norm{ M_{ij} \tensor \id_{ \cA \tensor \cB} }_\cY <
        1 - \frac{ \epsilon}{ 2} - \nu } \Big| \quad
    & > & \quad \beta n \enspace.
\end{IEEEeqnarray}

We pick bases~$\bB_i$ independently and uniformly at random, i.e., for
each~$i$, independently pick a Haar-random unitary operator
on~$\complex^m$, and let~$\bB_i$ be the basis defined by its columns.
Partition~$\bB_i$ into~$k$ sets~$\left(\bB_{ij} : j \in[k] \right)$ of
equal size. We then define an ensemble of the form in
Eq.~(\ref{eq-ensemble}) with~$\bm{\rho}_{ij} \coloneqq \tfrac{k}{m}
\sum_{\ket{v} \in \bB_{ij}} \density{v}$, and the corresponding projection
operators~$\bM_{ij} \coloneqq \sum_{\ket{v} \in \bB_{ij}} \density{v}$. 
We show that with non-zero probability, the 
operators~$\bM_{ij}$ satisfy Eq.~(\ref{eq-cond4}) for all~$\cY \in \fT$, 
by bounding the probability of the complementary event.

Suppose the operators~$\bM_{ij}$ do not satisfy Eq.~(\ref{eq-cond4}) for
some subspace~$\cY \in \fT$. Then
\begin{IEEEeqnarray}{rCl}
\label{eq-cond5}
\Big| \set{ (i,j) : \norm{ \bM_{ij} \tensor \id_{ \cA \tensor \cB} }_\cY <
        1 - \frac{ \epsilon}{ 2} - \nu } \Big| \quad
    & \le & \quad \beta n \enspace.
\end{IEEEeqnarray}
Equivalently, there are at least~$ (1 - \beta) n $ pairs~$i, j$ such 
that~$\norm{ \bM_{ij} \tensor \id }_\cY \ge 1 - \epsilon/2 - \nu $. In
particular, there are at least~$(1 - \beta) n / k$ indices~$i$ such that 
there is at least one~$j \in [k]$ with~$\norm{ \bM_{ij} \tensor \id
}_\cY \ge 1 - \epsilon/2 - \nu $.
For convenience, by~$\bE_i(\cY)$ we denote the event that there is
some~$j \in [k]$ with~$\norm{ \bM_{ij} \tensor \id
}_\cY \ge 1 -  \epsilon/2 - \nu $, and by~$\bI(\cY)$, we denote the subset
of indices~$i \in \left[ \tfrac{n}{k} \right]$ such that~$\bE_i(\cY)$ occurs.

Let~$q \coloneqq \ceil{(1 - \beta) \tfrac{n}{k}}$. By the above reasoning,
it suffices to 
bound the probability that for some subspace~$\cY \in \fT$, the
subset~$\bI(\cY)$ has at least~$q$ indices.

By Lemma~\ref{lemma-main}, for a fixed subspace~$\cY$ and pair~$i,j$,
\begin{IEEEeqnarray*}{rCl}
\Pr\Big[ \norm{\bM_{ij} \tensor \id}_{\cY} \geq 1 - \epsilon/2 - \nu 
        \Big] \quad
    & \leq & \quad \exp\left( - \frac{ ( (1 - \epsilon/2 - \nu) k
        - 2)^2 (m-2) }{768 k^2} \right) \\
    & \le & \quad \exp( - \gamma m ) \enspace,
\end{IEEEeqnarray*}
with~$\gamma \coloneqq \tfrac{(1 - \epsilon/2 - \nu)^2}{ 8 \times 768}$,
when~$(1 - \epsilon/2 - \nu) k \ge 4$ and
\[
m - 2 \quad \ge \quad 
    \frac{(16 \times 384) d}{(1 - \epsilon/2 - \nu)^2} \;
    \ln \! \left( \frac{8}{1 - \epsilon/2 - \nu} \right) 
    \enspace.
\]
So by the Union Bound
\begin{IEEEeqnarray*}{rCl}
\Pr \Big[ \bE_i(\cY) \Big] \quad & \le & \quad k \, \exp( - \gamma m ) \enspace,
\end{IEEEeqnarray*}
and by the Union Bound and the independence of~$\bM_{ij}$ for
distinct indices~$i$,
\begin{IEEEeqnarray*}{rCl}
\Pr \Big[ \size{ \bI(\cY) } \ge q \Big] \quad 
    & \le & \quad { { \frac{n}{k} } \choose q }
        \times \left( k \, \exp( - \gamma m ) \right)^q \enspace.
\end{IEEEeqnarray*}
Finally, we get
\begin{IEEEeqnarray*}{rCl}
\Pr \Big[ \exists \cY \in \fT : \textrm{Eq.~(\ref{eq-cond5}) holds}
        \Big] \quad
    & \le & \quad \size{\fT} \times \max_{\cY \in \fT} \: \Pr \Big[
        \size{ \bI(\cY) } \ge q \Big] \\
    & \le & \quad \left( \frac{ 8 \sqrt{d} }{ \nu } \right)^{ 2 d^2 m^2 }
        { \frac{n}{k} \choose q } 
        \left( k \, \exp( - \gamma m ) \right)^q \\
    & < & \quad 1 \enspace,
\end{IEEEeqnarray*}
when~$m > \max \left\{ \tfrac{3}{\gamma} \ln \left( \tfrac{\e}{ 1 -
\beta} \right) , \tfrac{3}{\gamma} \ln k \right\}$, and
\[
\gamma (1 - \beta) n \quad > \quad 6 k d^2 m \; \ln \! \left(
    \frac{8 \sqrt{d} }{ \nu } \right) \enspace.
\]
This proves the theorem.
\end{proof}

Note that the above proof considers an arbitrary choice of
states~$\sigma_{ij}$ and quantum channel~$\Psi$ \emph{after\/} the
ensemble is chosen randomly. Together, the sequence~$(\sigma_{ij})$ and
the channel~$\Psi$ constitute a compression 
protocol. The proof shows that no matter how~$(\sigma_{ij})$ and~$\Psi$ are
chosen, the error due to the corresponding compression protocol is
large if the dimension~$d$ is much smaller than~$m$ (provided~$n$ is
chosen properly).

\subsection{Application to entanglement cost}

Consider a one-way protocol~$\Pi$ in which with probability~$1/n$, Alice 
gets input~$(i,j)$, prepares state~$\rho_{ij}$ as in an ensemble given by 
Theorem~\ref{thm-compression}, and sends it to Bob.
The ensemble average~$\rho$ is 
the completely mixed state~$\tfrac{\id}{m}$ over~$\complex^m$. By construction, we 
have~$\rS( \rho_{ij} \| \rho) = \log k $, and therefore~$\QIC(\Pi) =
\tfrac{1}{2} \, \log k $. In fact, we
have~$\rS_{\max}( \rho_{ij} \| \rho) = \log k$. Theorem~I.1(1) of 
Ref.~\cite{BNR18} gives us a protocol for the visible compression of any
such ensemble of states using classical communication and shared
entanglement, with error~$\epsilon$. The communication cost of this
protocol is
\[
\Imax^{\epsilon/\sqrt{2}}(A:B)_{\tau} + \Order(\log \log (
    1/ \epsilon)) \enspace,
\]
where~$\tau^{AB} \coloneqq \tfrac{1}{n} \sum_{ij} \density{ij}^A \tensor
\rho_{ij}^B$ and we have used Proposition~\ref{prop-FvdG} to translate
between purified and trace distance. This expression is bounded from
above by~$\log k + \Order(\log \log \tfrac{1}{ \epsilon})$,
since~$\rS_{\max}( \rho_{ij} \| \rho)$ (and therefore~$\Imax(A:B)_{\tau}$) 
equals~$\log k$.  Using superdense
coding~\cite[Section~6.3.1]{W18-TQI}, we get a bound on the quantum
communication cost of compressing the ensemble with entanglement
assistance.
\begin{proposition}
\label{prop-ub}
For any positive integers~$k, m, n$ such that~$k$ divides~$m$ and~$n$,
and error parameter~$\epsilon > 0$,
any ensemble of~$n$ equally likely quantum states in~$\qstate(\complex^m)$ 
of the form in Eq.~(\ref{eq-ensemble}) there is a one-shot one-way
protocol \textbf{with shared entanglement} for compressing the states with
quantum communication at most
\[
\tfrac{1}{2} \log k + \Order(\log \log \tfrac{1}{ \epsilon}) \enspace,
\]
with average error at most~$\epsilon$ in trace distance.
\end{proposition}
This bound is an additive term of~$\Order(\log \log \tfrac{1}{ \epsilon})$
more than~$\QIC(\Pi)$. Theorem~I.1(1) in Ref.~\cite{BNR18} also gives a 
lower bound of~$(1/2) \, \Imax^{\sqrt{\epsilon}}(A:B)_{\tau}$ on the 
communication cost, which is at least~$(1/2) \log k - 2$ for~$\epsilon
\le 1/81$ (see Proposition~\ref{prop-imax-lb} in the appendix).
So for constant~$\epsilon$, the upper 
bound in Proposition~\ref{prop-ub} is close to optimal as a function of~$k$. 
It is slightly better than those obtained from protocols for state 
splitting (see, e.g., Ref.~\cite[Corollary~5]{AJ18-efficient-convex-split}), 
which have an additive term of order~$\log \tfrac{1}{\epsilon}$.
However, the protocol from Ref.~\cite{BNR18} has entanglement cost of 
order~$k (\log \tfrac{1}{\epsilon}) \log m$, which is exponential in 
the communication cost, while the protocol for state splitting with the 
least known communication cost~\cite[Corollary~5]{AJ18-efficient-convex-split}
has entanglement cost of order~$(1 + 1/\epsilon^2) \log (m / \epsilon)$.

Next we consider how small the entanglement cost of the visible compression 
of an ensemble~$( \rho_{ij})$ given by Theorem~\ref{thm-compression} may be.
By choosing the parameters in the statement of Theorem~\ref{thm-compression} 
appropriately, we get the following lower bound on the sum of 
communication and entanglement costs of any compression protocol.
\begin{corollary}
\label{cor-ent}
There exist universal constants~$c_1, c_2, c_3 >0$ such that for 
any~$\epsilon \in (0,1)$ and any positive integers~$k, m, n$ with~$m$ and~$n$
divisible by~$k$, there is an ensemble of~$n$ equally likely quantum states
in~$\qstate(\complex^m)$ of the form in Eq.~(\ref{eq-ensemble}) for
which any (one-shot) one-way protocol
for compressing the states with average error at
most~$\tfrac{\epsilon}{2}$, the sum of the communication and
entanglement costs is at least
\begin{equation}
\label{eq-ce-bd}
 \log m - 2 \log \frac{1}{1 - \epsilon } - \log \ln \frac{16}{ 1 -
\epsilon} - c_2 \enspace,
\end{equation}
when~$k \ge 6/(1-\epsilon)$, $m \ge c_1 (\ln k) / (1-\epsilon)^2$, and
\[
n \quad \ge \quad \frac{c_3}{ (1-\epsilon)^2 } \, k m^3 \ln
\frac{16 \sqrt{m}}{ \epsilon } \enspace.
\]
In particular, the entanglement cost 
of any such protocol with \textbf{optimal} communication cost is at 
least
\[
\log m - \frac{1}{2} \log k - \Order \Big(\log \frac{ 1}{1-\epsilon} \Big)
- \Order(1) \enspace,
\]
and the communication cost of any such protocol \textbf{without entanglement} 
is at least the bound given in Eq.~(\ref{eq-ce-bd}).
\end{corollary}
We defer the proof of this corollary to the appendix.

Note that the parameter~$m$ may be chosen arbitrarily larger than~$k$,
provided the number of states~$n$ in the ensemble is chosen large enough.
Thus, we see that there are ensembles with~$m$-dimensional
states for which communication-optimal compression protocols with shared
entanglement and with constant average error, say~$1/4$, have entanglement cost
almost as large as~$\log m$. In particular, the number of qubits of shared
entanglement needed may be arbitrarily larger than the quantum
information cost of the original protocol. We also see that in the 
\emph{absence\/} of shared entanglement, there are ensembles 
with~$m$-dimensional states that cannot be compressed to states with
dimension smaller than~$c m$ with average error less than~$1/4$,
where~$c$ is a universal positive constant. In particular, the
optimally compressed message may be arbitrarily longer than the quantum
information cost of the protocol~$\Pi$.

Corollary~\ref{cor-ent} shows that the number of qubits of shared entanglement 
used by protocol with the smallest known communication cost,
due to Anshu and Jain~\cite[Corollary~5]{AJ18-efficient-convex-split},
is optimal up to a constant multiplicative factor and
an additive~$\log k$ term (for constant error in compression).
The lower bound on entanglement cost given in the corollary may be achieved 
by protocols derived from those for state splitting, up to an additive
term of~$\tfrac{1}{2} \log k + \Order(1)$, again for constant 
error (see, e.g., Ref.~\cite[Lemma~3.3]{BCR11-reverse-Shannon}). However,
the communication cost of these protocols may not be optimal.

The probabilistic construction in the results above gives us ensembles 
with a number of states~$n$ that is polynomial in~$m$ and~$k$. Note that 
in the compression protocol~$\Pi'$, Alice 
may send the input~$(i,j)$ as her message, in which case the message 
register has dimension~$n$. Similarly, she may send the
state~$\rho_{ij}$ itself, and this has dimension~$m$.  So in order to 
study how much compression is truly possible (i.e., how much smaller the
dimension of the message register may be as compared with~$m$), we have 
to study ensembles with~$n \ge m$ states, and compression protocols with 
message registers with dimension at most~$m$. Further, consider any 
protocol~$\Gamma$ (similar to~$\Pi$) in which Alice receives a random
input~$x$ out of~$n$ possibilities according to some distribution,  
prepares a state~$\omega_x$ and sends it to Bob. The quantum information
cost of such a protocol~$\Gamma$ is at most~$\tfrac{1}{2} \log n$. 
So the polynomial dependence of~$n$ on the dimension of the 
states in the ensemble ($m$ in the construction above) and the exponential 
dependence of~$n$ on the quantum information cost of the corresponding 
protocol ($\tfrac{1}{2} \log k$ in the construction) is inevitable.

\section{Concluding remarks}
\label{sec-concl}

In this article, we revisited one-shot compression of an ensemble of quantum
states. We proved that there are ensembles which cannot be compressed by
more than a few qubits in the absence of shared entanglement, when
allowed constant error. In the presence of shared entanglement, the ensemble
can be compressed to many fewer qubits. However, the entanglement cost
may not be smaller than the number of qubits being compressed by more
than a constant, for constant error. 
Since we study compression protocols that are allowed to make some error, the bounds
we establish are robust to perturbations to the shared entangled state that 
are sufficiently small relative to the error.

Entanglement and quantum communication are distinct resources
in the context of information processing. Sharing 
entanglement involves the generation, distribution, and storage of a 
state that is independent of the input for the task at hand. Communication 
also involves the same steps, but may be dynamic, i.e., may depend on the input 
and the prior history of the communication protocol. Consequently, any
physical implementation of these resources is likely to incur different costs
for these steps. In this work, we focused on the cost of distributing quantum
states, and as a first stab, assumed that the cost of distribution for shared
entanglement or for communication is proportional to the number of qubits
involved. Formally, this corresponds to the notion of \emph{smooth~0-R{\'e}nyi entropy\/}.
The motivation for this focus comes largely from the area of communication complexity~\cite{LS09-communication-complexity}, in which the interaction
between multiple processors takes centre stage, but shared entanglement
is often taken for granted. Our result shows that entanglement plays a 
crucial role in important communication tasks and highlights the need for
considering entanglement cost in addition to communication cost.

A question of interest, from a theoretical perspective, is the 
\emph{degree\/} or \emph{strength\/} of entanglement required for different information 
processing tasks. Several different measures of entanglement have been studied
in the literature, depending on the context. Smooth~0-R{\'e}nyi entropy is a 
very coarse measure in this respect, as it may be the same for states that
are regarded as having widely different degrees of entanglement. A natural
question is whether results such as the ones we derived also hold for other
definitions of entanglement cost that capture the degree of entanglement more 
satisfactorily. We conjecture that analogous results hold also for other 
measures, and leave this to future work.

Many other questions surrounding compression remain open. For instance, we 
do not have tight characterizations for the communication and
entanglement costs of one-shot state re-distribution. Even lesser is 
known for the one-shot compression of interactive quantum protocols. 
Progress on these questions might hold the key to resolving important
questions in communication complexity as well.

\bibliographystyle{plainnat}
\bibliography{bibl}

\onecolumn\newpage
\appendix

\section{Proofs of some claims}
\label{sec-proofs}

In this section, we include the proofs of some statements from the main
body of the article.

\begin{proofof}{Corollary~\ref{cor-ent}}
We invoke Theorem~\ref{thm-compression} with~$\epsilon \in (0,1), ~ \nu = 
\epsilon / 2, ~ \beta = 1/2$ and~$k, m, n$ satisfying the conditions stated in the 
corollary. Then~$\gamma$ as in 
Theorem~\ref{thm-compression} equals~$(1 - \epsilon)^2 / (8 \times 768)$.
We take~$c_1 \coloneqq (24 \times 768) + 1$, so that~$m > (3 / \gamma) \ln k $.
Since~$k \ge 6/(1 - \epsilon)$, we have~$k > 6 > 2 \e = \e / (1 - \beta)$, 
and~$m > ( 3 / \gamma) \ln ( \e/ (1 - \beta))$. 
We take~$c_3 \coloneqq (6 \times 2 \times 8 \times 768) + 1$
so that~$n > (6 k m^3 / \gamma (1-\beta)) \ln (8 \sqrt{m} / \nu )$. 

Now we consider an ensemble~$( \rho_{ij})$ given by Theorem~\ref{thm-compression}.
Let~$\Pi'$ be any one-way protocol, possibly with shared entanglement, for 
the visible compression of the ensemble~$( \rho_{ij})$ with average error
at most~$\epsilon/2$. Following the notation from Section~\ref{sec-qcp}, suppose 
that Bob holds registers~$M E_\rB$ just after he receives the message~$M$ from 
Alice in~$\Pi'$. If the entanglement cost of~$\Pi'$ is~$e$, we may assume 
that the register~$E_\rB$ may be partitioned into sub-registers~$E_{1 \rB} 
E_{2 \rB}$ with~$\size{E_{1 \rB}} = e$, and that the state of
register~$E_\rB$ is of the form~$\omega \tensor \density{0}$, where~$E_{1 \rB}$
is in state~$\omega$ and~$E_{2 \rB}$ in state~$\density{0}$, and~$\ket{0}$
is a pure state. (We may achieve this by applying a suitable isometry 
to register~$E_\rB$.) 

Let~$d \eqdef \size{M E_{1 \rB}}$, so that the sum of the 
communication and entanglement costs of~$\Pi'$ is~$\log d$, and let~$\sigma_{ij}$ be 
the state of the registers~$M E_{1 \rB}$ with Bob when Alice is given 
input~$(i,j)$. If~$d \ge m$, the bound in Eq.~(\ref{eq-ce-bd}) holds, so
consider the case when~$d < m$. Then the choice of~$n$ above implies 
that~$n > (6 k d^2 m / \gamma (1-\beta)) \ln (8 \sqrt{d} / \nu )$.

Since the average error of~$\Pi'$ is at most~$\epsilon/2$, by the Markov
Inequality we have
\[
\Big| \set{ (i,j) : \trnorm{ \rho_{ij} - \Psi( \sigma_{ij} ) } > \epsilon } 
    \Big| \quad < \quad \frac{n}{2} \quad = \quad \beta n \enspace,
\]
where~$\Psi$ is the quantum channel corresponding to Bob's decompression
operation in~$\Pi'$. Theorem~\ref{thm-compression} then implies that
\[
2 + \frac{d}{ \gamma } \ln \! \left( \frac{16}{1 - \epsilon/2 - \nu} \right)
    \quad \ge \quad m  \enspace.
\]
Since~$m - 2 \ge m/2$, this gives us the bound stated in 
Eq.~(\ref{eq-ce-bd}) with~$c_2 \eqdef \log (16 \times 768)$.
\end{proofof}

\begin{proposition}
\label{prop-imax-lb}
Let~$(\rho_{ij})$ be an ensemble of the form in Eq.~(\ref{eq-ensemble}), and let the state~$\tau^{AB}$ be defined as~$\tfrac{1}{n} \sum_{ij} \density{ij}^A \tensor \rho_{ij}^B$. For any~$\zeta \in [0,1/8)$, we have
\[
\Imax^\zeta(A:B)_{\tau}  \quad \ge \quad \log k - \log \left( \frac{ 3 - 12 \zeta}{ 1 - 8 \zeta} \right) \enspace.
\]
\end{proposition}
\begin{proof}
As shown in Ref.~\cite[Proposition~II.5]{BNR18}, there is a classical-quantum state~$\tau'$ within purified distance~$\zeta$ of~$\tau$ such that~$\Imax^\zeta(A:B)_\tau = \Imax(A:B)_{\tau'}$. Let~$\tau' \eqdef \sum_{ij} q_{ij} \density{ij} \tensor \widetilde{\rho}_{ij}$.

By Proposition~\ref{prop-FvdG}, we have
\begin{equation}
\label{eq-dist}
\trnorm{ \tau - \tau' } \quad \le \quad 2 \zeta \enspace.
\end{equation}
Let~$\xi \eqdef 2 \zeta$. By monotonicity of trace distance under measurements~\cite[Proposition~3.5]{W18-TQI}, we further get
\[
\sum_{ij} \size{ q_{ij} - p_{ij} } \quad \le \quad \xi \enspace.
\]
If~$q_{ij} > 3/2n$ or~$q_{ij} < 1/2n$, we have~$\size{q_{ij} - p_{ij} } > 1/2n$.
So for at least~$(1 - 2 \xi)n$ pairs~$(i,j)$, we have~$1/2n \le q_{ij} \le 3/2n$, and 
we call such pairs~$(i,j)$ \emph{typical\/}.

Eq.~(\ref{eq-dist}) may be written as
\[
\sum_{ij} \trnorm{ q_{ij} \widetilde{\rho}_{ij} - p_{ij} \rho_{ij} } \quad \le \quad \xi \enspace,
\]
so, by monotonicity of trace distance,
\[
\sum_{ij} \sum_{ \ket{v} \in B_{ij}}
    \size{ q_{ij} \bra{v} \widetilde{\rho}_{ij} \ket{v} - \frac{k}{nm} } \quad \le \quad \xi \enspace,
\]
where~$B_{ij}$ is as in the definition of the ensemble~$(\rho_{ij})$.
In particular,
\begin{equation}
\label{eq-tr-dist}
\sum_{\text{typical } ij} ~~ \sum_{ \ket{v} \in B_{ij} }
        \size{ q_{ij} \bra{v} \widetilde{\rho}_{ij} \ket{v} - \frac{k}{nm} } \quad \le \quad \xi \enspace.
\end{equation}
There are at least~$(1 - 2 \xi)n/k$ indices~$i \in [n/k]$ such that there is a typical pair~$(i,j)$ for some~$j \in [k]$. Let~$S$ be the set of such indices~$i$. Let~$\eta \in (0,1)$. If for all indices~$i \in S$, there are less than~$(1 - \eta) m$ pairs~$(j,v)$ with~$(i,j)$ typical, $\ket{v} \in B_{ij}$, and
\begin{equation}
\label{eq-typical-ijv}
\frac{k}{2nm} \quad \le \quad  q_{ij} \bra{v} \widetilde{\rho}_{ij} \ket{v} \quad \le \quad \frac{3k}{2nm} \enspace,
\end{equation}
then we would have
\[
\sum_{\text{typical } ij} ~~ \sum_{ \ket{v} \in B_{ij} }
        \size{ q_{ij} \bra{v} \widetilde{\rho}_{ij} \ket{v} - \frac{k}{nm} }
    \quad > \quad (1 - 2 \xi) \frac{n}{k} \times \eta m \times \frac{k}{ 2nm}
    \quad = \quad (1 - 2 \xi) \frac{ \eta}{ 2} \enspace.
\]
Taking~$\eta \eqdef 2 \xi/(1 - 2 \xi)$, we see that this is in contradiction with Eq.~(\ref{eq-tr-dist}).
So there is an index~$i \in S$ such that there are at least~$(1 - \eta)m$ pairs~$(j,v)$ with~$j \in [k]$  and~$\ket{v} \in B_{ij}$ such that~$(i,j)$ is typical, and~$(i,j,v)$ satisfy Eq.~(\ref{eq-typical-ijv}).
Denote such an index~$i$ by~$i_0$, and let
\[
T \quad \eqdef \quad 
    \Big\{ (j,v) : j \in [k], ~ \ket{v} \in B_{i_0 j}, ~ (i_0, j) \text{ typical }, ~ (i,j,v) \text{ satisfy Eq.~(\ref{eq-typical-ijv}) } \Big\} \enspace.
\]

We have that for all the pairs~$(j,v) \in T$,
\[
\frac{k}{2nm} \quad \le \quad q_{i_0 j} \bra{v} \widetilde{\rho}_{i_0 j}\ket{v} 
    \quad \le \quad \frac{3}{2n} \bra{v} \widetilde{\rho}_{i_0 j}\ket{v} \enspace,
\]
so that
\begin{equation}
\label{eq-psd-lb}
\frac{k}{3m} \quad \le \quad \bra{v} \widetilde{\rho}_{i_0 j}\ket{v} \enspace.
\end{equation}

Let~$\sigma \in \qstate( \complex^m)$ be a state that achieves~$\Imax(A:B)_{\tau'}$, and let~$\lambda$ denote this max-information. For typical pairs~$(i,j)$, since~$q_{ij} > 0$, we have~$\widetilde{\rho}_{ij} \le 2^\lambda \sigma$. By Eq.~(\ref{eq-psd-lb}), we also have~$k/3m \le 2^\lambda \bra{v} \sigma \ket{v}$ for all pairs~$(j,v) \in T$. Summing up over all pairs~$(j,v) \in T$, we get~$(1 - \eta)k/3 \le 2^\lambda$, as the sets~$B_{i_0 j}$ are a partition of an orthonormal basis, and~$\sigma$ has trace at most~$1$. So~$\lambda \ge \log k - \log(3/(1 - \eta))$.
\end{proof}

\end{document}